\documentclass{article}

\usepackage{graphicx}
\usepackage{epstopdf}
\usepackage{amsmath, amsthm, latexsym}
\usepackage{amssymb}
\usepackage{slashed}
\usepackage{mathtools}

\DeclareMathOperator{\sech}{sech}

\usepackage{color}\usepackage{xcolor}

\usepackage[small, bf]{caption}
\usepackage{subfigure}
\usepackage{hyperref}

\newtheorem{proposition}{Proposition}
\newtheorem{remark}{Remark}

\usepackage{color}

\usepackage{parskip}
\numberwithin{equation}{section}

\usepackage{mathbbol}

\usepackage{braket}

\usepackage{secdot}
\usepackage{sectsty}
\sectionfont{\large}

\begin{document}

\vspace{0.5 cm}

\begin{center}
\textbf{\large Exact Solutions for Compactly Supported Parabolic and Landau Quantum Barriers}
\vspace{6mm}
\end{center}

\noindent \textbf{\normalsize Peter Collas}\footnote{Department of Physics and Astronomy, California
State University, Northridge, Northridge, CA 91330-8268. Email: peter.collas@csun.edu.}
\textbf{\normalsize and David Klein}\footnote{Department of Mathematics, California State University,
Northridge, Northridge, CA 91330-8313. Email: david.klein@csun.edu.}\\

\vspace{.2cm}
\parbox{11cm}{\noindent{\small Abstract. We derive exact solutions to the one-dimensional Schr\"odinger equation for compactly supported parabolic and hyperbolic secant potential barriers, along with combinations of these types of potential barriers.  We give the expressions for transmission and reflection coefficients and calculate some dwell times of interest.}} 
\vspace{.2cm}

\noindent {\small KEY WORDS:  Parabolic barriers, Landau barriers, quantum tunneling, dwell time}

\section{Introduction}\label{Intro}

Quantum tunneling plays far reaching roles in a broad range of physical phenomena arising in subatomic, atomic and molecular physics, condensed matter physics, microcircuitry, physical chemistry.  The 2025 Nobel prize for physics was awarded for research demonstrating macroscopic quantum mechanical tunneling in quantum Josephson junction circuits. \cite{DMC85,MDC85,MDC20}

Extensive investigations involving smooth barriers have been carried out, primarily using numerical methods \cite{DK10,MSSY14,H14,IL15,XH16}.  Research in parabolic potential barriers has been particularly active \cite{ST81,CH88,AGI95,HSADV13,LXLM20,Y24}.  In this paper, we find exact solutions to the one-dimensional Schr\"odinger equation for continuous parabolic potential barriers with compact support (see Fig. \ref{singleBplot}), and for the Landau and Lifshitz potential $U_0/\cosh^2 (\alpha x),\;\alpha, U_0 > 0$, as well as a modification of that potential with compact support (see Fig. \ref{LLUs}).  The potentials we consider are continuous in all cases, and it follows that the wave function solutions are at least twice continuously differentiable. This may be contrasted with rectangular barriers whose wave function solutions have discontinuous second derivatives.

Our paper is organized as follows.  In Section \ref{SinBar} we solve the Schr\"odinger equation for parabolic potential barriers with compact support.  In Section \ref{RT} we find exact expressions for the solution for an incident particle (from the left), including formulas for the transmission and reflection coefficients.  In Section \ref{MulBar} we consider multiple parabolic barriers and, in Section \ref{qbdw}, for a particular double parabolic barrier, we find a quasi-bound (resonant) state and compare the dwell times for the quasi-bound state and a typical regular state in several regions.  Then in Section \ref{LL} we consider in detail a barrier whose solution was first obtained by Landau and Lifshitz \cite{LL77} and correct an error in the calculation of this problem in ref. \cite{XH16}.  We go over the solution in detail and obtain wave function solutions for modified compact support versions of the potential.  Finally in Section \ref{sum} we summarize our results.  In the Appendices we provide several useful complements to the main text:  In Appendix \ref{ct} we give a transformation used in obtaining the parabolic barrier solution.  In Appendix \ref{psw} we give power series representations of the solutions.  In Appendix \ref{jdt} we review the concept of dwell time and include a proposition involving the use of the probability current $j_{in}$.  Since for a major part of our paper we adopted units for which $m=\hbar=1$,  in Appendix \ref{units} we explain in detail how to insert the $m$ and $\hbar$ back in the results.

\section{Parabolic barrier with compact support}\label{SinBar}

In this section, we consider the general symmetric parabolic barrier potential, conveniently, parametrized in terms of the positive constants $\alpha$ and $U_0$, given by
\begin{equation}
U(x)=
\begin{dcases}
0, &x\leq -\alpha\\
\frac{U_0}{\alpha^2}\left(-x^2+\alpha^2\right),&-\alpha\leq x\leq\alpha \label{U}\\
0, &x\geq\alpha\,.\\ 
\end{dcases}
\end{equation}
Observe that $U(0)=U_0$ and the support domain is $x\in[-\alpha,\alpha]$.  A plot is shown in Fig. \ref{singleBplot}. 
\begin{figure}[htbp!]
  \begin{center}
    \includegraphics[width=4 in]{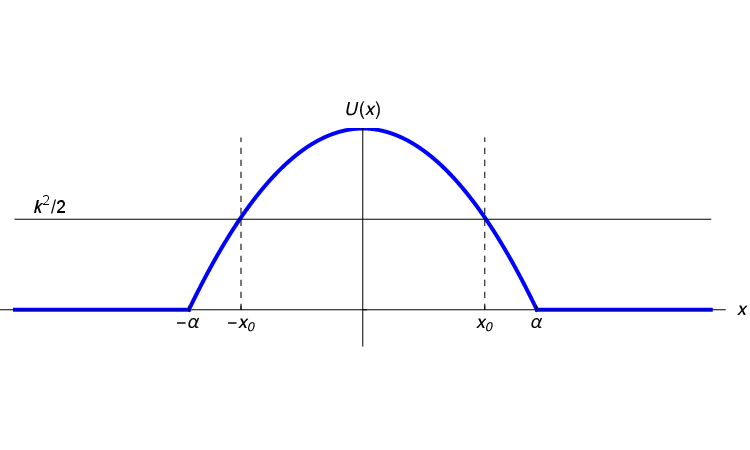}
     \end{center}
   \caption{The parabolic barrier $U(x)$ with $\alpha=1/2,\;U_0=1$.  The dashed vertical lines show the locations of the two turning points $x_0\approx 0.35$ for an incoming particle with kinetic energy $k^2/2=1/2$ with the convention $\hbar=m=1$.}
 \label{singleBplot}
\end{figure}

In what follows, we adopt units so that $\hbar=m=1$ (see Appendix \ref{units} for elaboration).  With this convention, the Schr\"{o}dinger equation restricted to the support interval $x\in\left[-\alpha,\alpha\right]$ for a particle with kinetic energy  $k^2/2$, is given by,
\begin{equation}
\frac{d^2 \psi(x)}{dx^2}+\left[\frac{2U_0}{\alpha^2}x^2 -(2U_0-k^2)\right]\psi(x)=0.\label{S1}
\end{equation}
  
To find the (exact) general solution to Eq. \eqref{S1}, we first consider the equation
\begin{equation}
\frac{d^2w(z)}{dz^2}+\left(\frac{z^2}{4}-a\right)w(z)=0.\label{wz}
\end{equation}
In Appendix \ref{ct}, we prove that if $w(z)$ is a solution of Eq.\eqref{wz}, then $V(x)=w(\sqrt{2}\,\sigma^{(1/4)}x)$ is a solution of
\begin{equation}
\frac{d^2V(x)}{dx^2}+\left(\sigma x^2-\lambda\right)V(x)=0,\label{Vx}
\end{equation}

with $\lambda=2a\sqrt{\sigma}$.  Equation \eqref{Vx} is identical to the Schr\"odinger equation \eqref{S1} provided
\begin{align}
\sigma&=\frac{2U_0}{\alpha^2},\label{sigma}\\\nonumber\\
\lambda&=2U_0-k^2.\label{lambda}
\end{align}
Solutions to Eq. \eqref{wz} are known and given in \cite{OLBC10}, Sec. 12, p. 315.  It is shown there that the even and odd functions $w_e$ and $w_o$ given by
\begin{align}
w_e(a,z)&=e^{-\frac{iz^2}{4}}M\left(\frac{1}{4}-\frac{ia}{2},\frac{1}{2},\frac{iz^2}{2}\right),\label{we}\\\nonumber\\
w_o(a,z)&=ze^{-\frac{iz^2}{4}}M\left(\frac{3}{4}-\frac{ia}{2},\frac{3}{2},\frac{iz^2}{2}\right),\label{wo}
\end{align}
are solutions to Eq. \eqref{wz}, where the functions $M$ are confluent hypergeometric functions. Our Appendix \ref{psw} shows why $w_e$ is even and $w_o$ is odd, and from this it follows that they are linearly independent, and thus any solution to Eq. \eqref{wz} is necessarily a linear combination of them.

Now let $V(x):=\psi(x)$ in Eq. \eqref{Vx} and define the independent solutions $\psi_e(x)$ and $\psi_o(x)$ to Eq.\eqref{S2} by,
\begin{equation}
\psi_e(x) = w_e(\sqrt{2\beta}\,x)\quad \text{and}\quad \psi_o(x)=w_o(\sqrt{2\beta}\,x),\label{wewo}
\end{equation}

where
\begin{equation}
\beta=\frac{\sqrt{2U_0}}{\alpha}.\label{beta}
\end{equation}

More precisely,
\begin{align}
\psi_e(\alpha,\beta,k,x)&=e^{-\frac{i\beta}{2}x^2}M\left[\frac{1}{4}\!\left(1+\frac{ik^2}{\beta}-i\alpha^2\beta\right),\frac{1}{2},\,i\beta x^2\right],\label{psie}\\\nonumber\\
\psi_o(\alpha,\beta,k,x)&=x\sqrt{2\beta}\,e^{-\frac{i\beta}{2}x^2}M\left[\frac{1}{4}\!\left(3+\frac{ik^2}{\beta}-i\alpha^2\beta\right),\frac{3}{2},\,i\beta x^2\right].\label{psio}
\end{align}

Using Eq. \eqref{windep} in Appendix \ref{psw} we see that
\begin{equation}
\psi_e(0)=1, \quad \psi_e^\prime(0)=0, \quad \psi_o(0)=0, \quad \psi_o^\prime(0)=\sqrt{2\beta}.\label{Fnorm}
\end{equation}

The above even and odd solutions are linearly independent, and therefore any solution of the Schr\"{o}dinger equation, Eq.\eqref{S1}, with $x$ restricted to the interval $[-\alpha,\alpha]$, is a linear combination of $\psi_e(x)$ and $\psi_o(x)$.

\section{Reflection and transmission coefficients}\label{RT}

In this section we follow an efficient approach described in Fl\"ugge \cite{F94} p.42. This approach is applicable to an arbitrary single symmetric barrier with compact support. In our case, $U(x)=U(-x)$ with support in the interval $x\in\left[-\alpha,\alpha\right]$ (c.f. Eq.\eqref{U}).  A solution to
\begin{equation}
\frac{d^2 \psi(x)}{dx^2}+\left[k^2-2U(x)\right]\psi(x)=0\label{S2}
\end{equation}

representing an incoming wave from the left, partially reflected and partially transmitted by the potential barrier given by Eq.\eqref{U}, has the form,
\begin{equation}
\psi(x)=\psi_{1}I_{\left[-\infty,-\alpha\right]}+\psi_{2}I_{\left[-\alpha,\alpha\right]}+\psi_{3}I_{\left[\alpha,\infty\right]},\label{psiI}
\end{equation}
where $I_{\left[a,b\right]}$ with $a<b$ is the indicator function taking the value $1$ in the interval $[a,b]$ and the value $0$ elsewhere, and 
\begin{align}
\psi_1&=e^{ikx} + re^{-ikx},\label{psi1}\\
\psi_2&=A \psi_e(x) + B \psi_o(x),\label{psi2}\\
\psi_3&=te^{ikx}.\label{psi3}
\end{align}

To insure that the function $\psi(x)$ is continuously differentiable at the points $x=-\alpha$ and $x=\alpha$, we must solve the following equations for $r, A, B, t$:
\begin{align}
e^{-ik\alpha} + re^{ik\alpha}&=\,\,\,\,A \psi_e(\alpha) - B \psi_o(\alpha),\label{sys1}\\
ik(e^{-ik\alpha} - re^{ik\alpha})&=-A\psi_e^\prime(\alpha) + B \psi_o^\prime(\alpha),\label{sys2}\\
te^{ik\alpha}&= \,\,\,\,A\psi_e(\alpha) + B \psi_o(\alpha),\label{sys3}\\
ikte^{ik\alpha}&= \,\,\,\,A\psi_e^\prime(\alpha) + B \psi_o^\prime(\alpha),\label{sys4}
\end{align}

where Eqs. \eqref{sys1}-\eqref{sys4}  are justified by the fact that $\psi_e$ is even and $\psi_o$ is odd.  Following Fl\"ugge, define
\begin{equation}
L_{\displaystyle e}:=\alpha\,\frac{\psi_e^\prime(\alpha)}{\psi_e(\alpha)}\quad \text{and} \quad L_{\displaystyle o}:=\alpha\,\frac{\psi_o^\prime(\alpha)}{\psi_o(\alpha)}. \label{L+L-}
\end{equation}

Adding Eq. \eqref{sys1} to Eq. \eqref{sys3} and subtracting Eq. \eqref{sys2} from Eq. \eqref{sys4} we obtain
\begin{align}
 (t+r)e^{ik\alpha}+e^{-ik\alpha}&=2A\psi_e(\alpha),\label{A1}\\
ik(t+r)e^{ik\alpha}-ike^{-ik\alpha}&=2A\psi^{\prime}_e(\alpha).\label{A2}
\end{align}

Taking their ratio and multiplying by $\alpha$ we obtain,
\begin{equation}
L_{\displaystyle e}=ik\alpha\,\frac{(t+r)e^{ik\alpha}-e^{-ik\alpha}}{(t+r)e^{ik\alpha}+e^{-ik\alpha}}.\label{L+}
\end{equation}

Subtracting Eq. \eqref{sys1} from Eq. \eqref{sys3} and adding Eq. \eqref{sys2} to Eq. \eqref{sys4}, we find
\begin{equation}
L_{\displaystyle o}=ik\alpha\,\frac{(t-r)e^{ik\alpha}+e^{-ik\alpha}}{(t-r)e^{ik\alpha}-e^{-ik\alpha}}.\label{L-}
\end{equation}

Solving Eqs. \eqref{L+} and \eqref{L-} for $r$ and $t$, we have
\begin{align} 
r=&-\frac{1}{2}e^{-2ik\alpha}\left[\frac{L_{\displaystyle e} + ik\alpha}{L_{\displaystyle e} - ik\alpha}+\frac{L_{\displaystyle o}+ ik\alpha}{L_{\displaystyle o} - ik\alpha}\right],\label{r}\\
t=&-\frac{1}{2}e^{-2ik\alpha}\left[\frac{L_{\displaystyle e} + ik\alpha}{L_{\displaystyle e} - ik\alpha}-\frac{L_{\displaystyle o} + ik\alpha}{L_{\displaystyle o} - ik\alpha}\right].\label{t}
\end{align}

Using the above results and solving for $A$ and $B$ from Eqs. \eqref{sys1} and \eqref{sys3} we obtain,
\begin{align}
A =& \frac{(t+r)e^{ik\alpha} +e^{-ik\alpha}}{2\psi_e(\alpha)},\label{A}\\
B =& \frac{(t-r)e^{ik\alpha} -e^{-ik\alpha}}{2\psi_o(\alpha)}.\label{B}
\end{align}

In addition, we can calculate the reflection coefficient $R$ and transmission coefficient $T$ from the above equations, 
\begin{equation}
R \equiv |r|^2=\frac{(L_{\displaystyle e} L_{\displaystyle o} + k^2\alpha^2)^2}{\left(L_{\displaystyle e}^2+k^2\alpha^2\right)
   \left(L_{\displaystyle o}^2+k^2\alpha^2\right)},\label{rpoly}
\end{equation}

and 
\begin{equation}
T\equiv |t|^2=\frac{k^2\alpha^2(L_{\displaystyle e} - L_{\displaystyle o})^2}{{\left(L_{\displaystyle e}^2+k^2\alpha^2\right)
   \left(L_{\displaystyle o}^2+k^2\alpha^2\right)}}.\label{tpoly}
\end{equation}

It follows from the first equalities in Eq.\eqref{rpoly} and \eqref{tpoly} that $R+T=1$.\\

\begin{remark}\label{C2}
We point out that $\psi(x)$ given by Eq. \eqref{psiI} is twice continuously differentiable at $x=\pm\alpha$ (and of course analytic at all other points).  To see this, first observe, using Eq.\eqref{S2}, that
\begin{equation}\label{LH1}
\lim_{x\to\alpha}\psi^{\prime\prime}(x)=\lim_{x\to\alpha}\left[-k^2+2U(x)\right]\psi(x)=\left[-k^2+2U(\alpha)\right]\psi(\alpha),
\end{equation}
by continuity of $\left[-k^2+2U(x)\right]\psi(x)$ at $x=\alpha$.  By definition,
\begin{equation}\label{LH2}
\psi^{\prime\prime}(\alpha)=\lim_{x\to\alpha}\frac{\psi^{\prime}(x)-\psi^{\prime}(\alpha)}{x-\alpha}.
\end{equation}
Applying L'H\^opital's rule to the right side of Eq.\eqref{LH2}, and using Eq.\eqref{LH1} yields,
\begin{equation}
\psi^{\prime\prime}(\alpha)=\lim_{x\to\alpha}\frac{\psi^{\prime}(x)-\psi^{\prime}(\alpha)}{x-\alpha}=\lim_{x\to\alpha}\psi^{\prime\prime}(x)=\left[-k^2+2U(\alpha)\right]\psi(\alpha)
\end{equation}
Thus, $\psi^{\prime\prime}(\alpha)$ exists and $\psi^{\prime\prime}(x)$ is continuous at $x=\alpha$ by Eq.\eqref{LH1}.  The same argument applies to the case $x=-\alpha$.  Moreover, this result holds for any continuous potential $U(x)$, including those considered in the following sections.
\end{remark}

\section{Multiple parabolic barriers}\label{MulBar}

In this section we consider multiple parabolic barriers. We begin with a translation along the $x$ axis of the potential Eq.\eqref{U} controlled by a parameter $\gamma$:
\begin{equation}
U(x)=
\begin{dcases}
0, &x\leq -\alpha+\gamma\\
\frac{U_0}{\alpha^2}\left(-(x-\gamma)^2+\alpha^2\right),&-\alpha+\gamma\leq x\leq\alpha+\gamma \label{U2}\\
0, &x\geq\alpha+\gamma.\\ 
\end{dcases}
\end{equation}

Using results from the previous sections, we can find exact wave function solutions for Eq. \eqref{U2} even though $U(x)\neq U(-x)$.

For $x\in\left[-\alpha+\gamma,\alpha+\gamma\right]$ the Schr\"odinger equation is,
\begin{equation}
\frac{d^2 \psi(x)}{dx^2}+\left[\frac{2U_0}{\alpha^2}(x-\gamma)^2 -(2U_0-k^2)\right]\psi(x)=0,\label{S3}
\end{equation}

and with this restriction on the range of $x$, the solutions, $\psi_e(x)$ and $\psi_o(x)$, are given by
\begin{align}
\psi_e(x)&=e^{-\frac{i\beta}{2}(x-\gamma)^2}M\left[\frac{1}{4}\!\left(1+\frac{ik^2}{\beta}-i\alpha^2\beta\right),\frac{1}{2},\,i\beta (x-\gamma)^2\right],\label{psie2}\\\nonumber\\
\psi_o(x)&=\sqrt{2\beta}(x-\gamma)e^{-\frac{i\beta}{2}(x-\gamma)^2}M\left[\frac{1}{4}\!\left(3+\frac{ik^2}{\beta}-i\alpha^2\beta\right),\frac{3}{2},\,i\beta (x-\gamma)^2\right],\label{psio2}
\end{align}

where again
\begin{equation}
\beta=\frac{\sqrt{2U_0}}{\alpha}.\label{beta1}
\end{equation}

Examples of double barriers are shown in Figs. \ref{doubleB} and \ref{doubleBsym}. The parameter values for the barrier on the left in Fig. \ref{doubleB} are given by $\alpha=1, U_0=1, \beta=\sqrt{2}, \gamma=-2$ and for the barrier on the right, $\alpha=1, U_0=2, \beta=2, \gamma=2$.  We elaborate on Fig. \ref{doubleBsym} below.

A more general potential has the form,
\begin{equation}
U=\sum_{i=1}^{n}U_{i}I_{\left[i\right]},\label{Ugen}
\end{equation}

where the union of the disjoint intervals (disjoint except possibly at endpoints) indexed by $i$ is the entire real line, and some of the $U_i$ may equal zero, while others have the form of Eq.\eqref{U2}.  The Schr\"odinger equation for potential $U$ may be written as
\begin{equation}
\psi^{\prime\prime}+\left(k^2-2U\right)\psi=0,\label{Ugen}
\end{equation}
where
\begin{equation}
\psi=\sum_{i=1}^{n}\psi_{i}I_{\left[i\right]}.\label{psigen}
\end{equation}

Then since
\begin{equation}
I_{\left[i\right]}I_{\left[i\right]}=I_{\left[i\right]}\;\;\text{and}\;\;I_{\left[i\right]}I_{\left[j\right]}=0\;\;\text{for}\;\; i\neq j,\label{Irules}
\end{equation}

the Schr\"odinger equation reduces to a sum of uncoupled differential equations
\begin{equation}
\sum_{i=1}^{n}\left[\psi_{i}^{\prime\prime}I_{\left[i\right]}+\left(k^2-2UI_{\left[i\right]}\right)\psi I_{\left[i\right]}=0\right].\label{Sgen}
\end{equation}

\begin{figure}[htbp!]
  \begin{center}
    \includegraphics[width=4 in]{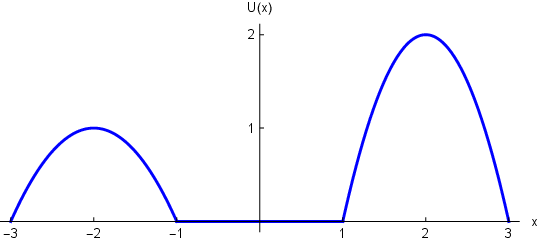}
     \end{center}
   \caption{A double parabolic barrier. For the parabola on the left, $U(\alpha,\gamma,U_0,x)=U(1,-2,1,x)$ and on the right $U(\alpha,\gamma,U_0,x)=U(1,2,2,x)$.}
 \label{doubleB}
\end{figure}
\begin{figure}[htbp!]
  \begin{center}
    \includegraphics[width=4 in]{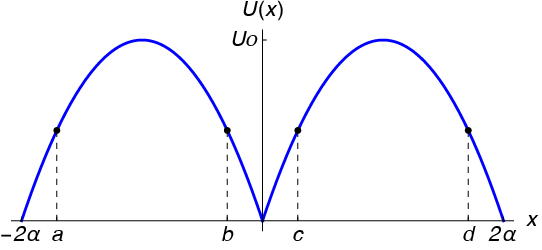}
     \end{center}
   \caption{A double barrier with turning points for a particle with energy  $E<U_0$ at $a, b, c, d$.}
 \label{doubleBsym}
\end{figure}
To further illustrate, we outline steps to calculate the wave function for the double barrier potential shown in Fig. \ref{doubleBsym}. For the barrier on the left $\gamma_1=-\alpha$ and for the barrier on the right $\gamma_2=\alpha$ so that $-\alpha+\gamma_1=-2\alpha$, $\alpha+\gamma_2=2\alpha$ and $\alpha+\gamma_1=-\alpha+\gamma_2=0$.  The turning points $a, b, c, d$ are the intersections of some energy line $E<U_0$ with the barrier.  In this double barrier case we have four regions so $\psi$ is
\begin{align}
\psi(x)&=\psi_{1}I_{\left[-\infty,-\alpha+\gamma_1\right]}+\psi_{2}I_{\left[-\alpha+\gamma_1,\alpha+\gamma_1\right]}+\psi_{3}I_{\left[-\alpha+\gamma_2,\alpha+\gamma_2\right]}+\psi_{4}I_{\left[\alpha+\gamma_2,\infty\right]}\nonumber\\
&:=\psi_{1}I_{\left[1\right]}+\psi_{2}I_{\left[2\right]}+\psi_{3}I_{\left[3\right]}+\psi_{4}I_{\left[4\right]}.\label{dpsiI1}
\end{align}

The wavefunctions in the regions 1 through 4 in Fig.\ref{doubleBsym} are given by,
\begin{align}
&\psi_1(k,x)=e^{ikx}+r\,e^{-ikx},\label{psi1}\\\nonumber\\
&\psi_2(\alpha,\beta,\gamma_1,k,x)=A\,\psi_e(\alpha,\beta,\gamma_1,k,x)+B\,\psi_o(\alpha,\beta,\gamma_1,k,x),\label{psi2}\\\nonumber\\
&\psi_3(\alpha,\beta,\gamma_2,k,x)=C\,\psi_e(\alpha,\beta,\gamma_2,k,x)+D\,\psi_o(\alpha,\beta,\gamma_2,k,x),\label{psi3}\\\nonumber\\
&\psi_4(k,x)=t\,e^{ikx},\label{psi4}
\end{align}
and, using the shorthand $\psi_i(x)\equiv\psi_i(\alpha,\beta,\gamma,k,x)$, the boundary conditions include
\begin{align}
&\psi_1(-\alpha+\gamma_1)=\psi_2(-\alpha+\gamma_1),\label{bpsi1}\\
&\psi_2(\alpha+\gamma_1)=\psi_3(-\alpha+\gamma_2),\label{bpsi2}\\
&\psi_3(\alpha+\gamma_2)=\psi_4(\alpha+\gamma_2).\label{bpsi3}
\end{align}

The additional three corresponding equations which match the derivatives of the $\psi_i$ at same positions are also required.  These six equations determine the six constants $A, B, C, D, r, t$ and it follows from the rules in Eq. \eqref{Irules} and Eq. \eqref{dpsiI1} that
\begin{equation}
|\psi|^2=|\psi_{1}|^2I_{\left[1\right]}+|\psi_{2}|^2I_{\left[2\right]}+|\psi_{3}|^2I_{\left[3\right]}+|\psi_{4}|^2I_{\left[4\right]}.\label{dpsi2}
\end{equation}
We will make use of this formula in the following section.

\section{Quasi-bound states and dwell times}\label{qbdw}

In this section, we find a quasi-bound state and calculate dwell times for the double parabolic barrier of Fig. \ref{doubleBsym} in various regions. The reader may wish to refer to Appendix \ref{jdt} where we review the concept of dwell time and prove a proposition \ref{Prop} regarding $j_{in}$ for the case of multiple barriers with compact support.

Here we assign numerical values obtained from Li and Yang \cite{LY24}, and for convenience converted to atomic units, as follows.  
\begin{align}
\hbar&\approx 10^{-34}\,\textup{J}\cdot \textup{s}=1\,\textup{au},\qquad\qquad\qquad\; m=m_e\approx 10^{-30}\,\textup{kg}=1\,\textup{aum},\label{F1SIa}\\
U_0&=3eV\approx 5\times 10^{-19}\,\textup{J}=0.125\,\textup{h}, \quad 2\alpha=10\AA\approx 10^{-9}\,\textup{m}=20\,\textup{b},\label{F1SIb}
\end{align}

where au stands for aomic units of action, aum stands atomic unit of mass, h stands for hartree and b for bohr, respectively.  Also in this section we use the notation of Eq. \eqref{mnot1} for our wave function, that is, $\varphi_i=\varphi_i(\hbar,m,\alpha,\gamma_1,\gamma_2,U_0,E,x)$. Analogous to Eq. \eqref{psigen}, we also write
\begin{equation}
\varphi=\sum_{i=1}^{4}\varphi_{i}I_{\left[i\right]},\label{psigen2}
\end{equation}
where the four intervals $\{I_{\left[i\right]}\}$ are the same as in Eq.\eqref{dpsiI1}.

Fig. \ref{doubleφ1} is a plot of $|\varphi|^2$ for the double barrier in Fig. \ref{doubleBsym} using the values in Eqs, \eqref{F1SIa} and \eqref{F1SIb} for the parameters and $\alpha=-\gamma_1$ for the barrier on the left and $\alpha=\gamma_2$ for the barrier on the right. The incident particle's energy is $E=0.02$h.  In the colored part of the curve the particle is under the influence of the potential.  The horizontal line along the $x$-axis is not zero and has more structure which cannot be seen without zooming in.

\begin{figure}[htbp!]
  \begin{center}
    \includegraphics[width=4 in]{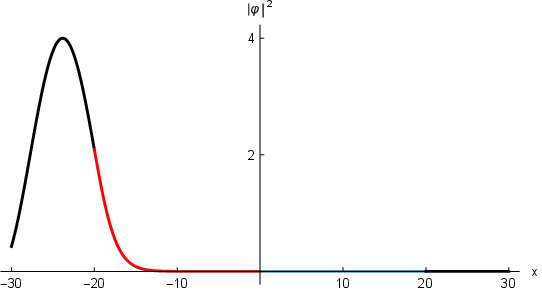}
     \end{center}
   \caption{We show a plot of $|\varphi|^2$ for the double barrier in Fig. \ref{doubleBsym} using the values in Eqs, \eqref{F1SIa} and \eqref{F1SIb} for the parameters and $\alpha=-\gamma_1$ for the barrier on the left and $\alpha=\gamma_2$ for the barrier on the right. The incident particle's energy is $E=0.02$h.}
 \label{doubleφ1}
\end{figure}

One may look for quasi-bound states by either finding the maximum of $|\varphi_2(E)|^2$ at $x=0$ or by numerically solving the equation $T(E)=1$ for $E$, where $T$ is the transmission coefficient.  In our case, for the parameter values given in Eqs, \eqref{F1SIa} and \eqref{F1SIb}, we find a quasi-bound state at $E=0.06115146$ h.  In Fig. \ref{doubleφ2} we see that plot of $|\varphi|^2$ is dramatically different from that in Fig. \ref{doubleφ1}.

\begin{figure}[htbp!]
  \begin{center}
    \includegraphics[width=4 in]{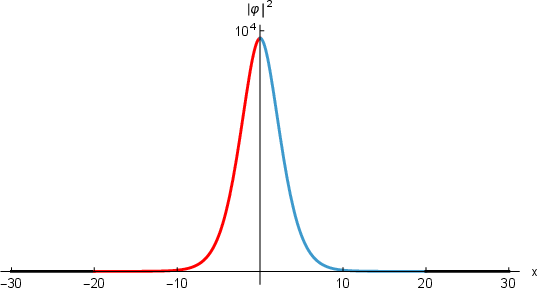}
     \end{center}
   \caption{This is again a plot of $|\varphi(x)|^2$ for the double barrier in Fig. \ref{doubleBsym} using the values in Eqs, \eqref{F1SIa} and \eqref{F1SIb} for the parameters and $\alpha=-\gamma_1$ for the barrier on the left and $\alpha=\gamma_2$ for the barrier on the right. The incident particle's energy is $E=0.06115146$ h which is the energy of the quasi-bound state.}
 \label{doubleφ2}
\end{figure}

We next compare some dwell times for the above two energy states.  For a given incoming particle with kinetic energy $E$, we let the four turning points be denoted by $a,\,b,\,c,\,d$, (see Fig. \ref{doubleBsym}).  The corresponding dwell times are given by
\begin{align}
\tau_{\left[a,b\right]}&=\frac{1}{j_{in}}\int_a^b\vert\varphi_2\vert^2dx,\label{taua}\\\nonumber\\
\tau_{\left[b,c\right]}&=\frac{1}{j_{in}}\int_b^0\vert\varphi_2\vert^2dx+\frac{1}{j_{in}}\int_0^c\vert\varphi_3\vert^2dx,\label{taub}\\\nonumber\\
\tau_{\left[c,d\right]}&=\frac{1}{j_{in}}\int_c^d\vert\varphi_3\vert^2dx.\label{tauc}
\end{align}

For a typical energy $E=0.02$ h, the turning points (in au) are, $a=-d,\,b=-c,\,c=0.835,\,d=19.165$, and $j_{in}=0.2$ au of velocity.  We then find
\begin{align}
\tau_{\left[a,b\right]}&=11.5\,\textup{aut},\label{tau1a}\\
\tau_{\left[b,c\right]}&=2.13\times 10^{-5}\,\textup{aut},\label{tau1b}\\
\tau_{\left[c,d\right]}&=1.40\times 10^{-5}\,\textup{aut},\label{tau1c}
\end{align}

where aut stands for au of time.

However at the energy $E=0.06115146$ h of the quasi-bound state the dwell times, as anticipated, change dramatically.  The turning points (in au) are, $a=-d,\,b=-c,\,c=2.85,\,d=17.1$, and $j_{in}=0.35$ au.  We then find
\begin{align}
\tau_{\left[a,b\right]}&=2.40\times 10^4\,\textup{aut},\label{tau2a}\\
\tau_{\left[b,c\right]}&=1.25\times 10^5\,\textup{aut},\label{tau2b}\\
\tau_{\left[c,d\right]}&=2.40\times 10^4\,\textup{aut}.\label{tau2c}
\end{align}

In addition to the order of magnitude changes, we point out the symmetry of the dwell times, which is again anticipated since we have that the transmission coefficient $T=1$ in this case.

\section{The Landau and Lifshitz barrier and its compact support version}\label{LL}

An interesting class of smooth single barriers is given by the potentials
\begin{equation}
U(\delta,x)=\frac{U_0}{\cosh^\delta(\alpha x)},\;\;\;\delta>0.\label{Udel}
\end{equation}

For $\delta$ a non-integer, the Schr\"odinger equation has an infinite number of branch points on the imaginary axis, while for integer $\delta$ it has an infinite number of irregular singularities, except when $\delta=1,2$, in which case the differential equation has an infinite number of regular singularities, namely, first and second order poles respectively.  

Landau and Lifshitz use the transformation, $\xi=\tanh{(\alpha x)}$, to transform the $\delta=2$ case Schr\"{o}dinger equation into the associated Legendre differential equation.  Unfortunately, in the case $\delta=1$, this transformation, as well as variations of it, give rise to square roots in the resulting differential equation. With the exception of $\delta=1$ or $2$, the resulting Schr\"{o}dinger equation has irregular singular points, so exact solutions for other values of $\delta$ are probably not available.  

The $\delta=2$ case is a problem in Landau and Lifshitz's quantum mechanics text \cite{LL77}, and more recently was solved independently by Xiao and Huang \cite{XH16}.  Unfortunately, due to an error in the transformation of their differential equation, their results for this barrier case are not correct.  Specifically a term is missing in the coefficient of $\psi(u)$ in their equation (32) and as a result, their final solution (in their notation) is for the potential $V(x,k)=(A-k^2/2)\sech^2{x/a}$, which now depends on $k$ also.  Consequently the resulting transmission coefficient, their Eq. (44), does not agree with the one given in our solution, Eq. \eqref{T<}, below, which is in agreement with  the transmission coefficient given in \cite{LL77}.

Not surprisingly Landau and Lifshitz provide only a terse outline of a solution, omitting several essential steps.  In this section we a provide a more detailed and direct derivation.  Then we modify the potential of Eq. \eqref{Udel} with $\delta=2$, to one with compact support, so that exact solutions to the Schr\"odinger equation with this modified potential, together with our parabolic or other potentials, can be found.  

The Landau and Lifshitz potential is
\begin{equation}
U(x)=\frac{U_0}{\cosh^2(\alpha x)},\label{LU2}
\end{equation}

and its graph is shown in Fig. \ref{barrier}.

\begin{figure}[htbp!]
  \begin{center}
    \includegraphics[width=4 in]{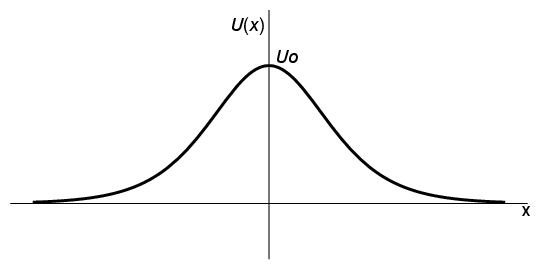}
     \end{center}
   \caption{Landau and Lifshitz's barrier.}
 \label{barrier}
\end{figure}

As in previous sections, we set $m=\hbar=1$.  The Schr\"{o}dinger equation for the potential of Eq. \eqref{LU2} is
\begin{equation}
\frac{d^2 \psi(x)}{dx^2}+2\left(E-\frac{U_0}{\cosh^2{(\alpha x)}}\right)\psi(x)=0.\label{LS1}
\end{equation}

It is easy to check that there is an \textit{irregular} singularity at $x=\infty$, where $\cosh^2{(\alpha x)}\rightarrow\infty,\, \sech^2{(\alpha x)}\rightarrow 0$ and the potential term vanishes.  There are also an infinite number of regular singularities at
\begin{equation}
x=\pm\, i\frac{(2n+1)}{\alpha}\frac{\pi}{2},\;\;n\in\mathbb{Z},\label{sing}
\end{equation}

these are 2nd order poles along the imaginary axis.

The coordinate transformation to the variable $\xi$, defined by
\begin{equation}
\sech^2{(\alpha x)}=\left(1-\xi^2\right),\;\;\;\xi\in\left[-1,1\right],\label{tr1}
\end{equation}

moves the singularity at $x=\infty$ to $\xi\pm1$, and since
\begin{equation}
\xi^2=1-\sech^2{(\alpha x)}=\tanh^2{\alpha x},\label{tr2}
\end{equation}

we choose the positive sign of the square root of Eq. \eqref{tr2} (see Remark \ref{symmetries} below)
\begin{equation}
\xi=\tanh{(\alpha x)}.\label{tr3}
\end{equation}

Moreover all the poles along the imaginary axis move to the point at infinity.  Using Eq. \eqref{tr3}, Eq. \eqref{LS1} is transformed into a differential equation with three \textit{regular} singular points at $(\pm1, \infty)$,
\begin{equation}
\left(1-\xi^2\right)\frac{d^{2}\psi}{d\xi^2}-2\xi\frac{d\psi}{d\xi}+\left[\frac{2E}{\alpha^2\left(1-\xi^2\right)}-\frac{2U_0}{\alpha^2}\right]\psi=0,\label{ALE}
\end{equation}

where now $\psi=\psi(\xi)$. Redefining the constants in the two terms inside the square brackets by,
\begin{align}
-\frac{2U_0}{\alpha^2}&=\nu(\nu+1),\label{nu}\\\nonumber\\
-\frac{2E}{\alpha^2}&=\mu^2,\label{mu}
\end{align}

and substituting these expressions into Eq. \eqref{ALE} results in a standard form of the Associated Legendre  equation, 
\begin{equation}
\left(1-\xi^2\right)\frac{d^{2}\psi}{d\xi^2}-2\xi\frac{d\psi}{d\xi}+\left[\nu(\nu+1)-\frac{\mu^2}{\left(1-\xi^2\right)}\right]\psi=0,\label{L1}
\end{equation}
whose solutions are the associated Legendre functions $\psi(\xi)=P^{\mu}_{\nu}(\xi)$.\\

\begin{remark}\label{symmetries}
The differential equation Eq. \eqref{L1} does not change under the substitutions $\mu\to -\mu,\;\nu\to -\nu-1$ or $\xi\to -\xi$. 
\end{remark}

\medskip

In our barrier case $E>0$, $k^2=2E$, $\mu^2=-k^2/\alpha^2<0$, and therefore $\mu=\pm ik/\alpha$.  From \cite{OLBC10}, p. 353, the Wronskian of $P^{\mu}_{\nu}(\xi)$ and $P^{-\mu}_{\nu}(\xi)$ is,
\begin{equation}
\mathcal{W}\left[P^{\mu}_{\nu}(\xi),P^{-\mu}_{\nu}(\xi)\right]=\frac{2\sin{(\mu \pi)}}{\pi\left(\xi^2-1\right)}.\label{W}
\end{equation}

We choose $\mu=ik/\alpha$, and since our $\mu$ is not an integer for $k\neq 0$, it follows that $P^{\mu}_{\nu}(\xi)$ and $P^{-\mu}_{\nu}(\xi)$ are two independent solutions.  As we shall see, $P^{\mu}_{\nu}(\xi)$ has the desired asymptotic behavior while $P^{-\mu}_{\nu}(\xi)$ does not.  Therefore we write the eigenfunction $\psi(\xi)$, corresponding to the energy eigenvalue $E=k^{2}/2$, as
\begin{equation}
\psi(\xi)=\mathcal{N}P^{\mu}_{\nu}(\xi),\label{psi}
\end{equation}

where the $k$-dependence is in $\mu$ and $\mathcal{N}$, an as yet undetermined normalization constant.
 
As $x\to \infty,\;\xi\to 1-$, the properly normalized $\psi(\xi)$, should be asymptotic to a transmitted wave $\psi_t$ moving to the right,
\begin{equation}
\psi\sim\psi_t=t\,e^{ikx},\label{a+}
\end{equation}

while, as $x\to -\infty,\;\xi\to -1+$, it should be asymptotic to an incident wave $\psi_i$ and a reflected $\psi_r$ wave,
\begin{equation}
\psi\sim\psi_i+\psi_r=e^{ikx}+r\,e^{-ikx},\label{a-}
\end{equation}

where $r^2=R$ and $t^2=T$ are the reflection and transmission coefficients respectively.  We show below that the solution of {Eq.\eqref{LU2} has the expected asymptotic behavior and we use it to find the exact expressions for $r$ and $t$ and hence $R$ and $T$.

In order to determine the asymptotic behavior of $\psi$ we shall need suitable expressions for $P^{\mu}_{\nu}(\xi)$ in terms of the hypergeometric functions $F\left[\alpha,\beta,\gamma,\tfrac{1}{2}(1\pm\xi)\right]$ valid in the interval $\xi\in (-1,1)$.

The standard expression of $P^{\mu}_{\nu}(\xi)$ in terms of $F$ in the interval $\xi\in (-1,1)$ is (see e.g., \cite{WG10}, p. 255),
\begin{equation}
P^{\mu}_{\nu}(\xi)=\frac{1}{\Gamma(1-\mu)}\left(\frac{1+\xi}{1-\xi}\right)^{\frac{\mu}{2}}F\left[-\nu,\nu+1,1-\mu,\tfrac{1}{2}(1-\xi)\right].\label{LP1}
\end{equation}

We shall see that in our case Eq. \eqref{LP1} is well-behaved as we approach the limit $x\to \infty,\;\xi\to 1-$,  since $F\left[-\nu,\nu+1,1-\mu,0\right]=1$ and the asymptotic behavior of $P^{\mu}_{\nu}(\xi)$ in that limit  is determined by the overall factors $(1\pm\xi)^{\pm\frac{\mu}{2}}$.  Unfortunately in the limit $x\to -\infty,\;\xi\to -1+$, we obtain $F\left[-\nu,\nu+1,1-\mu,1\right]$ and the hypergeometric series diverges for our set of parameters.

An alternative expression for $P^{\mu}_{\nu}(\xi)$, also valid in the interval $\xi\in (-1,1)$, and well-behaved in the limit $x\to -\infty,\;\xi\to -1+$, can be found by analytic continuation and is derived in e.g., \cite{WG10}, p. 257.  \begin{align}
P^{\mu}_{\nu}(\xi)&=\frac{\Gamma(-\mu)}{\Gamma(1+\nu-\mu)\Gamma(-\nu-\mu)}\left(\frac{1+\xi}{1-\xi}\right)^{\frac{\mu}{2}}F\left[-\nu,\nu+1,1+\mu,\tfrac{1}{2}(1+\xi)\right]\nonumber\\\nonumber\\
&-\frac{\sin{(\pi\nu)}}{\pi}\,\Gamma(\mu)\left(\frac{1-\xi}{1+\xi}\right)^{\frac{\mu}{2}}F\left[-\nu,\nu+1,1-\mu,\tfrac{1}{2}(1+\xi)\right].\label{LP2}
\end{align}
Note that in this case we have $F\left[-\nu,\nu+1,1\pm\mu,0\right]=1$, and again the asymptotic behavior is given by the overall factors of $(1\pm\xi)^{\pm\frac{\mu}{2}}$ of each term.

Using the relation $\xi=\tanh{(\alpha x)}$ and $\mu=ik/\alpha$ we find that
\begin{align}
\left(\frac{1+\xi}{1-\xi}\right)^{\frac{ik}{2\alpha}}&=e^{ikx},\label{asy1}\\\nonumber\\
\left(\frac{1-\xi}{1+\xi}\right)^{\frac{ik}{2\alpha}}&=e^{-ikx}.\label{asy2}
\end{align}

We now make use of the relations,
\begin{align}
\nu&=\frac{1}{2}\left(-1+\sqrt{1-\frac{8U_0}{\alpha^2}}\right),\label{s1}\\
\nu&(\nu+1)=-\frac{2U_0}{\alpha^2},\label{s2}
\end{align}

where $\nu$ could be real or complex, and Eq. \eqref{asy1} to evaluate Eq, \eqref{LP1} in the limit $x\to \infty,\;\xi\to 1-$.  We obtain the (un-normalized) transmitted part of the wavefunction
\begin{equation}
\psi\sim\frac{1}{\Gamma\left(1-\frac{ik}{\alpha}\right)}\,e^{ikx}.\label{psit}
\end{equation}

We repeat the process using both Eqs. \eqref{asy1} and \eqref{asy2} and evaluate Eq. \eqref{LP2}, in the limit $x\to -\infty,\;\xi\to -1+$.  In this case we obtain two terms, namely, the (un-normalized) incident and the reflected part of the wavefunction,
\begin{equation}
\psi\sim\frac{\Gamma\left(-\frac{ik}{\alpha}\right)}{\Gamma\left(1+\nu-\frac{ik}{\alpha}\right)\Gamma\left(-\nu-\frac{ik}{\alpha}\right)}\,e^{ikx}-\frac{\sin{(\pi\nu)}}{\pi}\,\Gamma\left(\frac{ik}{\alpha}\right)e^{-ikx}.\label{upsir}
\end{equation}

The three asymptotic parts of the wavefunction in Eqs. \eqref{psit} and \eqref{upsir} have to be normalized in accordance with Eq.\eqref{a-}, that is, the coefficient of $e^{ikx}$ in Eq. \eqref{upsir} should be equal to 1.  Dividing each term in Eqs. \eqref{psit} and \eqref{upsir}, by the coefficient of the first term of Eq. \eqref{upsir} gives us the required normalization factor, $\mathcal{N}$, in Eq. \eqref{psi},
\begin{equation}
\mathcal{N}=\frac{\Gamma\left(-\nu-\frac{ik}{\alpha}\right)\Gamma\left(1+\nu-\frac{ik}{\alpha}\right)}{\Gamma\left(-\frac{ik}{\alpha}\right)},\label{N}
\end{equation}

and consequently,
\begin{align}
r&=-\frac{\Gamma\left(-\nu-\frac{ik}{\alpha}\right)\Gamma\left(1+\nu-\frac{ik}{\alpha}\right)\Gamma\left(\frac{ik}{\alpha}\right)\sin{(\pi\nu)}}{\pi\;\Gamma\left(-\frac{ik}{\alpha}\right)},\label{Lr}\\\nonumber\\
t&=\frac{\Gamma\left(-\nu-\frac{ik}{\alpha}\right)\Gamma\left(1+\nu-\frac{ik}{\alpha}\right)}{\Gamma\left(1-\frac{ik}{\alpha}\right)\Gamma\left(-\frac{ik}{\alpha}\right)},\label{Lt}
\end{align}

and, as usual,
\begin{equation}
r^{*}r=R,\quad t^{*}t=T.\label{rt}
\end{equation}

\begin{remark}\label{UoIneq}
We note that the parameter $\nu$, Eq. \eqref{s1}, may be real or complex, depending on whether $(8U_0/\alpha^2)\lessgtr 1$ and consequently there will be two different $R$ and $T$'s depending on the numerical values of the parameters $U_0$ and $\alpha$ which affect the square root appearing in the denominator in Eq. \eqref{T<}.
\end{remark} 

The transmission coefficient for the case $(8U_0/\alpha^2)<1$ is
\begin{equation}
T=\frac{2 \sinh^2{\left(\frac{\pi k}{\alpha}\right)}}{\cos{\left(\pi\sqrt{1-\frac{8U_0}{\alpha^2}}\right)}+\cosh{\left(\frac{2 \pi k}{\alpha}\right)}},\label{T<}
\end{equation}

while in the case $(8U_0/\alpha^2)>1$ it is
\begin{equation}
T=\frac{2 \sinh^2{\left(\frac{\pi k}{\alpha}\right)}}{\cosh{\left(\pi\sqrt{\frac{8U_0}{\alpha^2}-1}\right)}+\cosh{\left(\frac{2 \pi k}{\alpha}\right)}}.\label{T>}
\end{equation}

We turn now to a modfication of the potential of Eq.\eqref{LU2}. It is straightforward to shift  Eq.\eqref{LU2} downward and/or horizontally by introducing the parameters $\beta$ and $\gamma$ as follows (see Fig. \ref{LLUs}).
\begin{equation}
U_s(x)=-\frac{\beta^2}{2}+\frac{U_0}{\cosh^2(\alpha (x-\gamma))}.\label{Lshift}
\end{equation}

\begin{figure}[htbp!]
  \begin{center}
    \includegraphics[width=4 in]{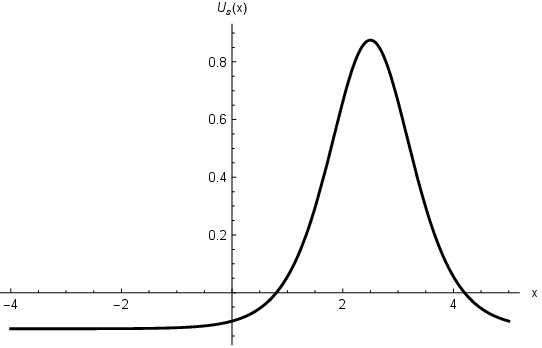}
     \end{center}
   \caption{The shifted Landau and Lifshitz's barrier.  The parameters are $U_0=1$, $\alpha=1$, $\beta=0.5$ and $\gamma=2.5$.}
 \label{LLUs}
\end{figure}

The potential $U_s(x)=0$ at 
\begin{equation}
x=\gamma\pm\frac{\cosh^{-1}{(\frac{\sqrt{2U_0}}{\beta})}}{\alpha},\label{Ush}
\end{equation}

provided that $\beta<\sqrt{2U_0}$.

The associated Legendre's equation for this potential is
\begin{equation}
\left(1-\xi^2\right)\frac{d^{2}\psi}{d\xi^2}-2\xi\frac{d\psi}{d\xi}+\left[\nu(\nu+1)-\frac{\mu^2}{\left(1-\xi^2\right)}\right]\psi=0,\label{LDE}
\end{equation}

where now
\begin{align}
\xi&=\tanh{(\alpha(x-\gamma))},\label{shxi}\\\nonumber\\
\mu&=\frac{i\sqrt{k^2+\beta^2}}{\alpha},\label{alphash}\\\nonumber\\
\nu&=\frac{1}{2}\left(-1+\sqrt{1-\frac{8U_0}{\alpha^2}}\right),\label{nush1}\\\nonumber\\
\nu&(\nu+1)=-\frac{2U_0}{\alpha^2}.\label{nush2}
\end{align}

This equation can be solved using the methods of this section, in particular for 
\begin{equation}
x\in \bigg[ \gamma\ - \frac{\cosh^{-1}{(\frac{\sqrt{2U_0}}{\beta})}}{\alpha},\,\, \gamma\ + \frac{\cosh^{-1}{(\frac{\sqrt{2U_0}}{\beta})}}{\alpha} \bigg],
\end{equation}
in other words, within the range of $x$ values where the potential is nonnegative.  

So we can "cut off" the part of the potential of Eq.\eqref{Lshift} that falls below the $x$ axis, thereby defining a new potential with compact support.  With this, one is able to combine various compact support potentials, of this type or parabolic, to create multiple barriers, see eg., Fig. \ref{mixedDBs}.\\

\begin{remark}\label{LLal}
We remind the reader that since we have adopted Landau and Lifshitz' notation in Sec. \ref{LL}, the $\alpha$ here has dimensions of inverse length and should not be confused with the $\alpha$ in the other sections.
\end{remark}

\begin{figure}[htbp!]
  \begin{center}
    \includegraphics[width=4 in]{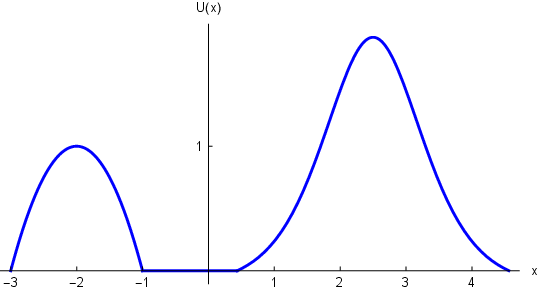}
     \end{center}
   \caption{An example of a ``mixed'' doubled barrier, parabolic on the left and a shifted potential of the form of Eq.\eqref{Lshift} on the right.}
 \label{mixedDBs}
\end{figure}

\newpage

\section{Summary}\label{sum}

We have derived exact wavefunction solutions to the one-dimensional Schr\"{o}dinger equation with parabolic potential barriers, Landau (i.e., squared hyperbolic secant) potential barriers and ``shifted'' versions of the latter with compact support.  Included among our results are exact calculations of transmission and reflection coefficients, dwell times, and identification of a quasi-bound state for a double parabolic potential. We showed how combining our results leads to exact solutions for mixed multiple barriers of the type studied here as well as other compactly supported barriers with known solutions.
 
\section*{Appendices}\label{App}
\appendix
\section{Coordinate transformation}\label{ct}

\begin{proposition}
Let $\sigma>0$, assume that $\lambda$ is real and that $V(x)$ is a solution of the differential equation
\begin{equation}
\frac{d^2 V}{dx^2}+ (\sigma x^2 - \lambda)V=0.\label{1}
\end{equation}

Then the function 
\begin{equation}
w(z):=V\left(\frac{z}{\sqrt{2}\,\sigma^{(1/4)}}\right),\label{2}
\end{equation}

is a solution to
\begin{equation}
\frac{d^2 w}{dz^2}+\left(\frac{z^2}{4}-\frac{\lambda}{2\sqrt{\sigma}}\right)\!w(z)=0.\label{3}
\end{equation}

Conversely, if $w(z)$ is a solution to
\begin{equation}
\frac{d^2 w}{dz^2}+\left(\frac{z^2}{4}-a\right)w(z)=0,\label{4}
\end{equation}

where the parameter $a$ is real, then the function
\begin{equation}
V(x)=w(\sqrt{2}\,\sigma^{(1/4)}x),\label{5}
\end{equation}

is a solution to
\begin{equation}
\frac{d^2 V}{dx^2}+ (\sigma x^2 - 2a\sqrt{\sigma})V=0.\label{6}
\end{equation}
\end{proposition}

\begin{proof}
Let $x$ be a function of $z$ given by
\begin{equation}
x(z)= \frac{z}{\sqrt{2}\,\sigma^{(1/4)}}\label{7}
\end{equation}

and define,
\begin{equation}
w(z):=V(x(z))= V\left(\frac{z}{\sqrt{2}\,\sigma^{(1/4)}}\right),\label{8}
\end{equation}

So we may regard $V=V(x(z))$ as a function of $z$ and the same is true of the second derivative of $V$ with respect to $x$. 

then by the chain rule,
\begin{equation}
\frac{dw}{dz}=\frac{dV}{dx}\frac{dx}{dz}=\frac{dV}{dx}\frac{1}{\sqrt{2}\,\sigma^{(1/4)}}\label{9}
\end{equation}

or,
\begin{equation}
\frac{dV}{dx}=\sqrt{2}\,\sigma^{(1/4)}\frac{dw}{dz}.\label{10}
\end{equation}

Therefore,
\begin{equation}
\frac{d^2w}{dz^2}=\frac{d}{dz}\left(\frac{dw}{dz}\right)=\frac{d}{dx}\left(\frac{dV}{dx}\frac{1}{\sqrt{2}\,\sigma^{(1/4)}}\right)\frac{dx}{dz}.\label{11}
\end{equation}

Simplifying Eq.\eqref{11} gives,
\begin{equation}
\frac{d^2 V}{dx^2}=2\sqrt{\sigma}\,\frac{d^2 w}{dz^2}\label{12}.
\end{equation}

Combining the above results, we have that
\begin{equation}
\frac{d^2 V}{dx^2}+(\sigma x^2 - \lambda)V= 2\sqrt{\sigma}\frac{d^2 w}{dz^2}+\left(\sigma \frac{z^2}{2\sqrt{\sigma}}-\lambda\right)\!w(z)=0,\label{13}
\end{equation}

therefore
\begin{equation}
\frac{d^2 w}{dz^2}+\left(\frac{z^2}{4}-\frac{\lambda}{2\sqrt{\sigma}}\right)\!w(z)=0.\label{14}
\end{equation}

Thus if $V$ is a solution of Eq. \eqref{6}, then $w(z):= V\!\left(z/(\sqrt{2}\,\sigma^{(1/4)})\right)$ is a solution of Eq. \eqref{4} with $a=\lambda/(2\sqrt{\sigma})$.

In a similar way, one can go in the reverse direction.

Let $z$ be a function of $x$ given by
\begin{equation}
z(x)=\sqrt{2}\,\sigma^{(1/4)}x,\label{15}
\end{equation}

and define,
\begin{equation}
V(x)=w(z(x))=w(\sqrt{2}\,\sigma^{(1/4)}x).\label{16}
\end{equation}

Then by the chain rule,
\begin{equation}
\frac{dV}{dx}=\frac{dw}{dz}\frac{dz}{dx}=\sqrt{2}\,\sigma^{(1/4)}\frac{dw}{dz},\label{17}
\end{equation}

or,
\begin{equation}
\frac{dw}{dz}=\frac{1}{\sqrt{2}\,\sigma^{(1/4)}}\frac{dV}{dx}.\label{18}
\end{equation}

Similarly,
\begin{equation}
\frac{d^2V}{dx^2}=\sqrt{2}\,\sigma^{(1/4)}\frac{d}{dx}\left(\frac{dw}{dz}\right)=\sqrt{2}\,\sigma^{(1/4)}\frac{d}{dz}\left(\frac{dw}{dz}\right)\frac{dz}{dx}= 2\sqrt{\sigma}\,\frac{d^2w}{dz^2}.\label{19}
\end{equation}

Thus,
\begin{equation}
\frac{d^2w}{dz^2}=\frac{1}{2\sqrt{\sigma}}\frac{d^2V}{dx^2}.\label{20}
\end{equation}

Therefore,
\begin{equation}
\frac{d^2 w}{dz^2}+\left(\frac{z^2}{4}-a\right)w(z)=\frac{1}{2\sqrt{\sigma}}\frac{d^2V}{dx^2}+\left(\frac{2\sqrt{\sigma}}{4}x^2 -a\right)\!V=0.\label{21}
\end{equation}

Simplifying gives,
\begin{equation}
\frac{d^2 V}{dx^2}+ (\sigma x^2 - 2a\sqrt{\sigma})V=0.\label{22}
\end{equation}

Thus if $w$ is a solution of Eq. \eqref{3}, then $V(x)= w(z(x))=w(\sqrt{2}\sigma^{(1/4)}x)$ is a solution of Eq. \eqref{6} with $\lambda=2a\sqrt{\sigma}$.
\end{proof}

\section{Power series representation of $w_e$ and $w_o$}\label{psw}

The functions $w_e(a,z)$ and $w_o(a,z)$ given respectively by Eqs. \eqref{we} and \eqref{wo} may also be expressed in terms of power series (cf. \cite{OLBC10}) as follows: 

\begin{align}
w_e(a,z)&=\sum^{\infty}_{n=0}\alpha_n \frac{z^{2n}}{(2n)!},\label{eseries}\\
w_o(a,z)&=\sum^{\infty}_{n=0}\beta_n \frac{z^{2n+1}}{(2n+1)!}\label{oseries},
\end{align}

In these series, the coefficients, $\alpha_n$ and  $\beta_n$, satisfy these recursion relations:
\begin{align}
\alpha_{n+2}&= a\,\alpha_{n+1}-\tfrac{1}{2}(n+1)(2n+1)\alpha_n,\label{rec1}\\
\beta_{n+2}&= a\,\beta_{n+1}-\tfrac{1}{2}(n+1)(2n+3)\beta_n,\label{rec2}
\end{align}

with $\alpha_0 = 1, \alpha_1 = a, \beta_0 = 1, \beta_1 = a$.

It follows immediately from Eqs \eqref{eseries} and \eqref{oseries} that $w_e(a,z)$ and $w_o(a,z)$ are real valued and that $w_e(a,z)$ is even and $w_o(a,z)$ is odd.  Moreover,
\begin{equation}\label{windep}
w_e(a, 0)=1, \quad w_e^\prime(a,0)=0, \quad w_o(a,0)=0, \quad w_o^\prime(a,0)=1,
\end{equation}
from which Eqs.\eqref{Fnorm} follow.

\section{Dwell time and probability current}\label{jdt}

Most authors attribute the standard equation for the dwell time to B\"{u}ttiker \cite{Bu83}.  He considered incident particles on a single rectangular barrier of height $V$ and extending from $x=0$ to $x=a$. The wave function solution, in this case, consists of the usual three parts:
\begin{align}
\psi_1(x)&=e^{ikx}+r\,e^{-ikx},\;\;\;x\leq 0,\label{cpsi1}\\
\psi_2(x)&=b\,e^{\kappa x}+c\,e^{-\kappa x},\;\;\;0\leq x\leq a,\label{cpsi2}\\
\psi_3(x)&=t\,e^{ikx},\hspace{1.65cm}a\leq x.\label{cpsi3}
\end{align}

With the above ``normalization'' of $\psi_{1}$, B\"{u}ttiker gave the expression below for the dwell time $\tau$,
\begin{equation}
\tau=\left(\frac{1}{j_{in}}\right)\int_{0}^{a}\vert\psi\vert^{2}dx=\left(\frac{1}{k}\right)\int_{0}^{a}\vert\psi\vert^{2}dx,\label{Bdt}
\end{equation}

where the incoming current $j_{in}=k$.   Winful \cite{W03} defines $\tau$ as a measure of the time spent by a particle in the barrier region $x\in\left[0,a\right]$ regardless of whether the particle is ultimately transmitted or reflected.

We note that in B\"{u}ttiker \cite{Bu83}, $\tau$ is defined as the ratio of the number of particles in $\left[0,a\right]$, with energy $E=k^2/2$, to the incident flux $j_{in}=k$.

Leavens and Aers \cite{LA89} reconcile the time dependent wave packet treatment with the steady-state scattering solution of the time-independent Schr\"{o}dinger equation and give a more general definition of $\tau$.

We also mention that the dwell times of bound states e.g., the states of the particle in the box or, say, the bound states of $U(x)=-U_{0}/\cosh^{2}{(\alpha x)}$, are infinite since $j_{in}=0$ in these cases.

Our considerations below will apply to one-dimensional potentials $V(x)$ with compact support in some interval $x\in\left[a,c\right]$, with particles going from left to right.

Although a large number of papers have dealt with and applied Eq. \eqref{Bdt} in different situations,  \cite{W03,LA89,LA89b,HS89,W06} and references therein, there is no proof that one should use $j_{in}=k$ for every region of a potential $V(x)$ with more than one barrier, since the initial incoming wave, moving to the right, has altered forms in the regions between barriers. We supply the missing proof below.

We refer to Fig. \ref{regDB} for a concrete illustration.  Such potentials have been considered by Dutt and Kar \cite{DK10}.\\

\begin{proposition}\label{Prop}
Let  $V(x)$ be a one-dimensional potential (in general a smooth multiple barrier) with compact support in $x\in\left[a,c\right]$.  Assume that the particles are moving from left to right with fixed momentum.  For region 1, defined by $x<a$, let the incoming wavefunction be normalized so that $\psi=e^{ikx}+re^{-ikx}$, while in region 7, $x>c$, we have $\psi=te^{ikx}$, for some values of $r$ and $t$.  Then the equation for the dwell time in the interval $\left[x_1,x_2\right]$, for any $x_1$ and $x_2$ with $x_1 < x_2$, is
\begin{equation}
\tau_{\left[x_{1}x_{2}\right]}=\left(\frac{1}{k}\right)\int_{x_{1}}^{x_{2}}\vert\psi(x)\vert^{2}dx.\label{dt1}
\end{equation}
\end{proposition}

\begin{proof}
The dwell time $\tau_{\left[a,c\right]}$ of a particle of energy $E$ in the interval $\left[a,c\right]$, is given by
\begin{equation}
\tau_{\left[a,c\right]}=\left(\frac{1}{j_{in}}\right)\int_{a}^{c}\vert\psi(x)\vert^{2}dx=\left(\frac{1}{k}\right)\int_{a}^{c}\vert\psi(x)\vert^{2}dx,\label{DT}
\end{equation}

where $\psi(x)$ is the eigenfunction with energy $E=k^{2}/2$ and $j_{in}=k$, the incoming probability current entering $V(x)$ at $x=a$.

For simplicity and to fix ideas, we refer to Fig. \ref{regDB} for the rest of the proof and we suppose that $x_1$ and $x_2$ are consecutive turning points. We may write
\begin{equation}
\psi(x)=\psi_{1}(x)I_{\left[1\right]}(x)+\psi_2(x)I_{\left[2\right]}(x)+\cdots+\psi_7(x)I_{\left[7\right]}(x),\label{psi}
\end{equation}

\begin{figure}[htbp!]
  \begin{center}
    \includegraphics[width=4 in]{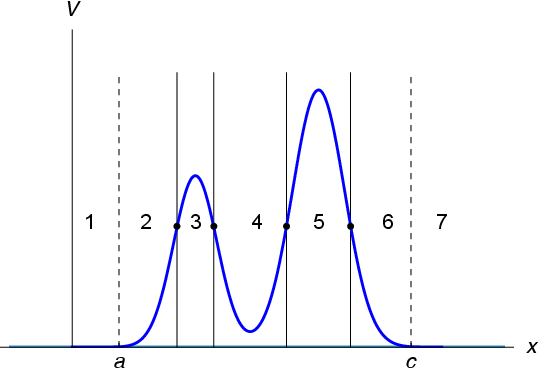}
     \end{center}
   \caption{A smooth double barrier and its regions.  The particle is moving from left to right.  The regions are determined by the intersections, of the horizontal line with ordinate the incident particle's energy $k^{2}/2$, with $V(x)$.}
  \label{regDB}
\end{figure}

where the 7 regions in Fig. \ref{regDB} are determined by the intersections of the horizontal line with ordinate equal to the incident particle's energy $k^{2}/2$, with $V(x)$.  The black dots show a hypothetical set of such intersections.  Each $I_{\left[A\right]}(x)$ is the indicator function for region $A$.  It is clear that
\begin{equation}
\vert\psi\vert^2=\vert\psi_1\vert^2I_{\left[1\right]}+\vert\psi_2\vert^2I_{\left[2\right]}+\cdots+
\vert\psi_7\vert^2I_{\left[7\right]},\label{psisq}
\end{equation}

since $I_{\left[A\right]}I_{\left[A\right]}=I_{\left[A\right]}$, and $I_{\left[A\right]}I_{\left[B\right]}=0$ if the intervals $\left[A\right]$ and $\left[B\right]$ overlap in at most one point.
Therefore Eq. \eqref{DT} becomes
\begin{align}
\tau_{\left[a,c\right]}&=\left(\frac{1}{k}\right)\int_{a}^{c}\sum_{A=1}^{5}\vert\psi_{A}\vert^{2}I_{\left[A\right]}dx\\\nonumber\\
&=\sum_{A=1}^5\tau_{\left[A\right]},\label{dts}
\end{align}

where
\begin{equation}
\tau_{\left[A\right]}=\left(\frac{1}{k}\right)\int_{A}\vert\psi_{A}\vert^{2}dx.\label{fdt}
\end{equation}
\end{proof}

\section{Inserting the $m$ and $\hbar$ back in the solutions}\label{units}

In performing calculations it is often convenient to set $\hbar=m=1$, as we have done throughout this article.  However, for the purpose of comparing calculations with experimental results, it may be necessary to recover these terms in the final expressions that result from calculations.

In our case, we wish to restore the appearance of  $\hbar$ and $m$ in the wave functions $\psi_i$ such as in Eqs \eqref{psie}, \eqref{psio}, \eqref{psie2}, \eqref{psio2} and so on, as well as for dwell times for various potentials and regions.

The wave function solutions, $\psi_i$, depend on position $x$ as well as parameters $\alpha,U_0,\gamma,k$, but it has been more convenient for us to use the parameter $\beta$ (see Eq.\eqref{beta}) in place of $U_0$. So we write $\psi_i=\psi_i(\alpha,\beta,\gamma,k,x)$, and recall that $\alpha$ and $\gamma$ are lengths and $k=p/\hbar$, where $p$ is momentum.

To carry this out, one must substitute $k=\sqrt{2mE}/\hbar$, and it follows from the Schr\"odinger equation,
\begin{equation}
\label{schroed}
\psi^{\prime\prime}-\frac{2m}{\hbar^2}(E-U)\psi=0,
\end{equation}
 that  $U_0$ must be replaced by $mU_0/\hbar^2$, in which case $\beta=\sqrt{2mU_0}/(\hbar\alpha)$. Then, when two parameters $\gamma_1$ and $\gamma_2$ are involved, for example,  the wave function becomes
\begin{equation}
\label{mnot1}
\varphi_i(\hbar, m, \alpha, \gamma_1,\gamma_2, U_0, E, x)\equiv\psi_i\left(\alpha, \frac{\sqrt{2mU_0}}{\alpha\hbar}, \gamma_1,\gamma_2, \frac{\sqrt{2mE}}{\hbar},x\right).
\end{equation}

The same substitutions must be made in formulas for dimensionless constants in the solutions, e.g., $r,\,A,\,B,\,C,\,D,\,t,$ in Eqs. \eqref{psi1}-\eqref{psi4}.

The incoming current $j_{in}$ and the turning points $x_t$ have dimensions, so they also require substitutions. The incoming current $j_{in}$ is evaluated using
\begin{equation}
\label{j}
j(\psi)=\frac{i\hbar}{2m}\left(\psi\partial_x\psi^\ast-\psi^\ast\partial_x\psi\right),
\end{equation}
and is then given by
\begin{equation}
\label{jinmh}
j\left[e^{i\frac{\sqrt{2mE}}{\hbar}x}\right]=j(m,E)=\sqrt{\frac{2E}{m}}.
\end{equation}

From Eq.\eqref{schroed}, the turning points, $x_t$, are found by solving the equation $U=E$.  For the  cases of Figs. \ref{singleBplot} and \ref{doubleBsym}, $U$ is given by Eq. \eqref{U2}.  So from $U=E$ we find that
\begin{equation}
\label{xt}
x_t(\alpha, \gamma, U_0, E)=\gamma\pm\alpha\sqrt{\left(1-\frac{E}{U_0}\right)}.
\end{equation}
It follows that the formula for dwell time is
\begin{equation}
\label{tauanyU}
\tau_{\left[a,b\right]}=\sqrt{\frac{m}{2E}}\int_{a}^{b}\vert\varphi_i\vert^2dx,
\end{equation}
which has dimensions of time, as it should.

Having introduced the $m$ and $\hbar$ in the expressions of interest, we may evaluate quantities in any system of units.  We have made use of the international system of units (SI) and the atomic system of units (au) in Sec. \ref{qbdw}.\\

\begin{remark}\label{m=h}
Note that in atomic units $m=\hbar=1$, but one may adopt other systems which also have $m=\hbar=1$, but where the unit of energy is not the hartree and the unit of length is not the bohr, but rather some other convenient reference scales.
\end{remark}

\end{document}